\documentclass[a4paper,UKenglish,cleveref, autoref]{lipics-v2019}

\title{Longest Common Subsequence on Weighted Sequences}

\titlerunning{Longest Common Subsequence on Weighted Sequences}

\author{Evangelos Kipouridis}{Basic Algorithms Research Copenhagen (BARC), University of Copenhagen, Denmark}{kipouridis@di.ku.dk}{https://orcid.org/0000-0002-5830-5830}{Thorup’s Investigator Grant 16582, Basic Algorithms Research Copenhagen (BARC), from the VILLUM Foundation,
and European Union’s Horizon 2020 research and innovation program under the Marie Skłodowska-Curie grant agreement No 801199. \includegraphics[height=0.03\textwidth, width=0.11\textwidth]{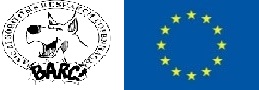}}

\author{Kostas Tsichlas}{School of Informatics, Aristotle University of Thessaloniki, Greece}{tsichlas@csd.auth.gr}{}{}

\authorrunning{E. Kipouridis and K. Tsichlas}

\Copyright{Evangelos Kipouridis and Kostas Tsichlas}

\ccsdesc[100]{Theory of computation~Approximation algorithms analysis}
\ccsdesc[100]{Theory of computation~W hierarchy}
\ccsdesc[100]{Theory of computation~Problems, reductions and completeness}

\keywords{WLCS, LCS, weighted sequences, approximation algorithms, lower bound}

\category{}

\nolinenumbers 

\acknowledgements{We would like to thank the anonymous reviewers for their careful reading of our paper and their many insightful comments and suggestions.}

\hideLIPIcs  

\EventEditors{Inge Li G{\o}rtz and Oren Weimann}
\EventNoEds{2}
\EventLongTitle{31th Annual Symposium on Combinatorial Pattern Matching (CPM 2020)}
\EventShortTitle{CPM 2020}
\EventAcronym{CPM}
\EventYear{2020}
\EventDate{June 17--19, 2020}
\EventLocation{Copenhagen, Denmark}
\EventLogo{}
\SeriesVolume{161}
\ArticleNo{21}

\begin{document}

\maketitle

\begin{abstract}
We consider the general problem of the Longest Common Subsequence ($LCS$) on weighted sequences. Weighted sequences are an extension of classical strings, where in each position every letter of the alphabet may occur with some probability. Previous results presented a $PTAS$ and noticed that no $FPTAS$ is possible unless $P=NP$. In this paper we essentially close the gap between upper and lower bounds by improving both. First of all, we provide an $EPTAS$ for bounded alphabets (which is the most natural case), and prove that there does not exist any $EPTAS$ for unbounded alphabets unless $FPT=W[1]$. Furthermore, under the Exponential Time Hypothesis, we provide a lower bound which shows that no significantly better $PTAS$ can exist for unbounded alphabets. As a side note, we prove that it is sufficient to work with only one threshold in the general variant of the problem.
\end{abstract}

\section{Introduction} \label{sec:intro}

\subsection{General concepts}
We consider the problem of determining the $LCS$ (Longest Common Subsequence) on weighted sequences. Weighted sequences, also known as $p$-weighted sequences or Position Weighted Matrices (PWM) \cite{DBLP:journals/jda/AmirGS10,journals/nar/ThompsonHG94} are probabilistic sequences which extend the notion of strings, in the sense that in each position there is some probability for each letter of an alphabet $\Sigma$ to occur there.

Weighted sequences were introduced as a tool for motif discovery and local alignment and are extensively used in molecular biology \cite{DBLP:books/cu/Gusfield1997}. They have been studied both in the context of short sequences (binding sites, sequences resulting from multiple alignment, etc.) and on large sequences, such as complete chromosome sequences that have been obtained using a whole-genome shotgun strategy \cite{DBLP:conf/soda/Myers00,DBLP:conf/recomb/Venter02}. Weighted sequences are able to keep all the information produced by such strategies, while classical strings impose restrictions that oversimplify the original data.

Basic concepts concerning the combinatorics of weighted sequences (like pattern matching, repeats discovery and cover computation) were studied using weighted suffix trees \cite{DBLP:conf/ifipTCS/IliopoulosMPPTT04}, Crochemore's partitioning \cite{DBLP:journals/almob/BartonIP14a,DBLP:journals/algorithmica/BartonP18,DBLP:journals/jcb/ChristodoulakisIMPTT06}, the Karp-Miller-Rabin algorithm \cite{DBLP:journals/jcb/ChristodoulakisIMPTT06}, and other approaches \cite{DBLP:journals/bmcbi/ZhangGI13, DBLP:journals/mst/KociumakaPR19}. Other interesting results include approximate and gapped pattern matching \cite{DBLP:conf/cpm/AmirIKP06,DBLP:conf/ciac/ZhangGI10, DBLP:journals/iandc/RadoszewskiS20}, online pattern matching \cite{DBLP:journals/iandc/CharalampopoulosIPR19}, weighted indexing \cite{DBLP:journals/tcs/AmirCIKZ08, DBLP:journals/iandc/BartonKLPR20}, swapped matching \cite{DBLP:conf/cis/ZhangGI04}, the all-covers and all-seeds problem \cite{journals/papl/ZhangGFI10,DBLP:conf/aaim/ZhangGI10}, extracting motifs \cite{DBLP:conf/spire/IliopoulosPTTT04}, and the weighted shortest common supersequence problem \cite{DBLP:conf/spire/AmirGS11, DBLP:conf/spire/CharalampopoulosKPRRSWZ19}. There are also some more practical results on mapping short weighted sequences to a reference genome \cite{DBLP:journals/ijcbdd/AntoniouIMP09} (also studied in the parallel setting \cite{DBLP:journals/ijfcs/IliopoulosMP12}), as well as on the reporting version of the problem which we also consider in this paper \cite{DBLP:journals/algorithmica/BartonP18}.

The Longest Common Subsequence ($LCS$) problem is a well-known measure of similarity between two strings. Given two strings, the output should be the length of the longest subsequence common to both strings. Dynamic programming solutions \cite{DBLP:journals/cacm/Hirschberg75,DBLP:journals/jacm/WagnerF74} for this problem are classical textbook algorithms in Computer Science. $LCS$ has been applied in computational biology for measuring the commonality of DNA molecules or proteins which may yield similar functionality. A very interesting survey on algorithms for the $LCS$ can be found in \cite{DBLP:conf/spire/BergrothHR00}. The current $LCS$ algorithms are considered optimal, since matching lower bounds (under the Strong Exponential Time Hypothesis) were proven \cite{DBLP:conf/focs/AbboudBW15,DBLP:conf/soda/BringmannK18}.

Extensions of this problem on more general structures have also been investigated (trees and matrices \cite{DBLP:journals/tcs/AmirHKST08}, run-length encoded strings \cite{DBLP:journals/jc/ApostolicoLS99}, and more). One interesting variant of the $LCS$ is the Heaviest Common Subsequence ($HCS$) where the matching between different letters is assigned a different weight, and the goal is to maximize the weight of the common subsequence, rather than its length.

\subsection{Weighted LCS}
The problem studied in this paper is the weighted $LCS$ (WLCS) problem. It was introduced by Amir et al. \cite{DBLP:journals/jda/AmirGS10} as an extension of the classical $LCS$ problem on weighted sequences. Given two weighted sequences, the goal is to find a longest string which has a high probability of appearing in both sequences. Amir et al. initially solved an easier version of this problem in polynomial time, but unfortunately its applications are limited. As far as the general problem is concerned, they hinted its NP-Hardness by giving an NP-Hardness result on a closely related problem, which they call the log-probability version of WLCS. In short, the problem is the same, but all products in its definition are replaced with sums. Their proof is based on a Turing reduction and only works for unbounded alphabets. Finally, Amir et al. provide an $\frac{1}{|\Sigma|}$-approximation algorithm for the WLCS problem.

Cygan et al. \cite{DBLP:journals/dam/CyganKRRW16} strengthened the evidence that WLCS is NP-Hard by providing an NP-Completeness result on the decision log-probability version of WLCS (informally introduced in the previous paragraph), already for alphabets of size $2$, using a Karp reduction; for alphabets of size $1$ the solution is trivial since there is no uncertainty. They also gave an $\frac{1}{2}$-approximation algorithm and a $PTAS$, while also noticing that an $FPTAS$ cannot exist, assuming WLCS is indeed NP-Hard, as hinted by their evidence, and that P $\neq$ NP. Finally, they proved that every instance of the problem can be reduced to a more restricted class of instances. However, for this to be achieved their algorithm needs to perform exact computations of roots and logarithms that may make the algorithm to err.

Finally, it is worth noting that Charalampopoulos et al. \cite{DBLP:conf/spire/CharalampopoulosKPRRSWZ19}, proved that unless P=NP, WLCS cannot be solved in $\mathcal{O}(n^{f(a)})$ time, for any function $f(a)$, where $a$ is the cut-off probability. We note that this result concerns exact computations rather than approximations.

\subsection{Our results}
In this paper we essentially close the gap between upper and lower bounds for WLCS by improving both; we prove that the problem is indeed NP-Hard even for alphabets of size $2$. Furthermore, we provide an $EPTAS$ for bounded alphabets. These two results, along with the $FPTAS$ observation by Cygan et al. completely characterize the complexity of WLCS for bounded alphabets.
For unbounded alphabets, a $PTAS$ was already known by Cygan et al. \cite{DBLP:journals/dam/CyganKRRW16}. We show matching lower bounds, both by ruling out the possibility of an $EPTAS$, and by showing that, under the Exponential Time Hypothesis, no significantly better $PTAS$ can exist.
We also prove that every instance of WLCS can be reduced to a restricted class of instances without using roots and logarithms, thus being able to actually achieve exact computations without rounding errors that can make the algorithm err.

As noted in the previous paragraph, apart from essentially closing the gap between hardness results and faster algorithms we also circumvent the need to work with roots and logarithms as the previous results did. In short, by taking advantage of the property that $(ab)^c=a^cb^c$ and setting $c$ to be an appropriate logarithm, previous results transformed any instance to a more manageable form. However, this transformation introduces an error that may make the algorithm err as noted in Appendix~\ref{app_sec:one_threshold}.
Table~\ref{tab:Results} summarizes the above discussion. Table~\ref{tab:AlphabetSize} summarizes our results depending on the alphabet-size.

A short discussion is in order with respect to what new insights on weighted $LCS$ enabled us to achieve progress. Our most crucial observation is the fact that the problem behaves differently in the natural case of a bounded alphabet, and in the case of an unbounded alphabet. Without this distinction, closing the gap between upper and lower bounds was unlikely. That's because, on the one hand, no $EPTAS$ for the general case could be found, as none existed. On the other hand, proving that no $EPTAS$ exists requires reductions that work only on unbounded alphabets. The aforementioned distinction is what enabled us to understand that modifying the existing reductions, which work for alphabets of size $2$, would be futile in proving $W[1]$-Hardness.

Furthermore, it was crucial to identify that working with products is the core difficulty in proving NP-Hardness of weighted $LCS$. The introduction of the log-probability version of the weighted $LCS$ reflects the assumption that the difference between working with sums and working with products is just a technicality. After \cite{DBLP:journals/jda/AmirGS10} and \cite{DBLP:journals/dam/CyganKRRW16} successfully proved NP-Hardness for the log-probability version, it was natural to attempt reducing from it for proving NP-Hardness of the weighted $LCS$ problem. Despite the apparent similarities between the two problems, their difference did not allow us to craft such a reduction. For the same reason, Cygan et al. used a model that assumed infinite precision computations with reals, while we are able to avoid such a strong assumption.

\subsection{Organization of the paper}
The rest of the paper is organized as follows. In Section~\ref{sec:preliminaries}, we provide necessary definitions and discuss the model of computation. In Section~\ref{sec:NP}, we show that WLCS is NP-Complete while in Section~\ref{sec:eptas}, we provide the $EPTAS$ algorithm for bounded alphabets, which is also an improved $PTAS$ for unbounded alphabets. In Section~\ref{sec:noeptas}, we show that there can be no $EPTAS$ for unbounded alphabets by showing that this problem is $W[1]$-hard and in Section~\ref{sec:lowerbound}, we describe the matching conditional lower bound. We conclude in Section~\ref{sec:conclusion}.

For clarity purposes, some proofs and technical discussions are moved to the Appendix. More specifically, in Appendix~\ref{app_sec:one_threshold} we present an algorithm that transforms any instance of our problem to an equivalent, but easier to handle, instance. We also show that the rounding errors introduced by working with reals (logarithms and roots) may cause a similar algorithm by Cygan et al. \cite{DBLP:journals/dam/CyganKRRW16} to err if standard rounding is used.

\begin{table}[ht]
\caption{Results on WLCS.}
\noindent\makebox[\textwidth]{
\label{tab:Results}
\begin{tabular}{|l|l|l|l|}
\hline
                                                                                             & Amir et al.                                                      & Cygan et al.                                                                             & Our results                                                                                        \\ \hline
NP-Hardness of WLCS                                                                           & \begin{tabular}[c]{@{}l@{}}Hinted, by NP-Hardness of\\ Log-probability version\\ (Turing reduction -\\ only for unbounded alphabets)\end{tabular} & \begin{tabular}[c]{@{}l@{}}Hinted, by NP-Hardness of\\ Log-probability version\\ (Karp reduction - \\ already from alphabets size $2$)\end{tabular} &  \begin{tabular}[c]{@{}l@{}}Proved \\ (Karp Reduction - \\ already from alphabets of size $2$) \end{tabular} \\ \hline
Approximation Algorithms                                                                     & $\frac{1}{\Sigma}$-Approximation & $PTAS$                                                                                     & \begin{tabular}[c]{@{}l@{}}$EPTAS$ for bounded\\ alphabets,\\ Improved $PTAS$ for\\ unbounded\end{tabular} \\ \hline
\begin{tabular}[c]{@{}l@{}}Proof that no $EPTAS$ exists\\ for unbounded alphabets\end{tabular} & No                                                               & No                                                                                       & Yes                                                                                                \\ \hline
Lower bound on any $PTAS$                                                                          & No                                                               & No                                                                                       & \begin{tabular}[c]{@{}l@{}}Matching the\\upper bound,\\under $ETH$\end{tabular} \\ \hline
\begin{tabular}[c]{@{}l@{}}Reduction to a restricted\\ class of instances\end{tabular}       & No                                                               & \begin{tabular}[c]{@{}l@{}}Yes, by assuming exact\\ computations of logarithms \end{tabular} & \begin{tabular}[c]{@{}l@{}}Yes, without any \\ extra assumptions \end{tabular} \\ \hline
\end{tabular}
}
\end{table}

\begin{table}[ht]
\caption{Results depending on the Alphabet Size}
\noindent\makebox[\textwidth]{
\label{tab:AlphabetSize}
\begin{tabular}{|l|l|l|ll}
\cline{1-3}
Alphabet Size & Previous Results                                                           & Our results                                                                                        &  &  \\ \cline{1-3}
$1$             & Trivial                                                                    & Trivial                                                                                            &  &  \\ \cline{1-3}
Constant Size      & No $FPTAS$ possible & Achieved $EPTAS$                                                                                     &  &  \\ \cline{1-3}
Depending on the input          & Achieved $PTAS$                                                              & \begin{tabular}[c]{@{}l@{}}No $EPTAS$ possible,\\ Improved $PTAS$, \\ Matching Lower Bound\end{tabular} &  &  \\ \cline{1-3}
\end{tabular}
}
\end{table}

\section{Preliminaries} \label{sec:preliminaries}
\subsection{Basic Definitions}
Let $\Sigma = \{\sigma _1, \sigma _2, \ldots, \sigma _K\}$ be a finite alphabet. We deal both with bounded $(K=O(1))$ and unbounded alphabets. $\Sigma ^d$ denotes the set of all words of length $d$ over $\Sigma$. $\Sigma ^*$ denotes the set of all words over $\Sigma$.
\begin{definition}[Weighted Sequence]
A weighted sequence $X$ is a sequence of functions $p^{(X)}_1,\ldots,p^{(X)}_{|X|}$, where each function assigns a probability to each letter from $\Sigma$. We thus have $\sum_{j=1}^K{p^{(X)}_i(\sigma _j)}=1$ for all $i$, and $p^{(X)}_i(\sigma _j)\geq 0$ for all $i,j$.
\end{definition}

By $WS(\Sigma)$ we denote the set of all weighted sequences over $\Sigma$. Let $X\in WS(\Sigma)$. Let $Seq^{|X|}_d$ be the set of all increasing sequences of $d$ positions in $X$. For a string $s\in \Sigma^d$ and $\pi \in Seq^{|X|}_d$, define $P_X(\pi , s)$ as the
probability that the subsequence on positions corresponding to $\pi$ in $X$ equals $s$.
More formally, if $\pi =(i_1, i_2, \ldots, i_d)$ and $s_k$ denotes the $k$-th letter of $s$, then
\[P_X(\pi , s)= \prod_{k=1}^d{p^{(X)}_{i_k}(s_k)}\]

Denote \[SUBS(X,a)=\{s\in \Sigma ^* \vert \exists \pi \in Seq^{|X|}_{|s|} ~such~that~P_X(\pi , s)\geq a\}\]
That is, $SUBS(X,a)$ is the set of deterministic strings which match a subsequence of $X$ with probability at least $a$. Every $s\in SUBS(X,a)$ is called an $a$-subsequence of $X$.

Let us give a clarifying example. If $\Sigma=\{\sigma_1,\sigma_2\}$ and $X$ is a long weighted sequence, where in each position the probability for each letter to appear is $0.5$, then $SUBS(X,0.3)$ does not contain $s=\sigma_1\sigma_1$, as, for any increasing subsequence of $2$ positions, the probability of $s$ appearing is $0.25<0.3$.

The decision problem we consider is the following:
\begin{definition}[$(a_1,a_2)$-WLCS decision problem]
Given two weighted sequences $X, Y$, two cut-off probabilities $a_1, a_2$ and a number $k$, find if the longest string $s$ contained in $SUBS(X,a_1)\cap SUBS(Y,a_2)$ has length at least $k$.
\end{definition}

Naturally, the respective optimization problem is the following:
\begin{definition}[$(a_1,a_2)$-WLCS optimization problem]
Given two weighted sequences $X, Y$, and two cut-off probabilities $a_1, a_2$, find the length of the longest string contained in $SUBS(X,a_1)\cap SUBS(Y,a_2)$.
\end{definition}

Both in the decision and the optimization version, the WLCS problem is the $(a_1,a_2)$-WLCS problem, where $a_1=a_2$. We denote these (equal) probabilities by $a$ ($a=a_1=a_2$) for concreteness. 

Let us note that the problem is only interesting if $|\Sigma|\geq 2$. For $|\Sigma|=1$ the problem is trivial since there is no uncertainty at all. The same letter appears in every position in both strings with probability $1$, and thus the answer is simply the length of the shorter weighted sequence.

Finally, let us also state that the Log-Probability version of the WLCS, studied in previous papers, is the same as the original WLCS if we replace $P_X(\pi , s)= \prod_{k=1}^d{p^{(X)}_{i_k}(s_k)}$ by $P_X(\pi , s)= \sum_{k=1}^d{p^{(X)}_{i_k}(s_k)}$.

\subsection{Model of Computation}
Our model of computation is the standard word $RAM$, introduced by Fredman and Willard \cite{DBLP:conf/stoc/FredmanW90} to simulate programming languages like C. The word size is $w=\Omega(\log{I})$, where $I$ is the input size in bits, so as to allow random access indexing of the whole input. Thus, arithmetic operations between words take constant time.
However, due to the nature of our problem, it is necessary to compute products of many numbers. This can produce numbers that are much larger than the word size. We even allow numbers in the input to be larger than $2^w$ (these numbers just need to use more than one word to be represented). We generally assume that each number in the input is represented by at most $B$ bits, but we do not pose any constraint on $B$ other than the trivial one that $B<I$.
Of course, in cases where we deal with numbers that occupy many words, we no longer have unit-cost arithmetic operations; we guarantee, however, that our results only use linear or near-linear time operations (like comparisons and multiplications) on numbers polynomial in the input size. Thus, although we do not enjoy the unit-cost assumption for arbitrary numbers, we always stay within the polynomial-time regime.

\subsection{Basic Operations}
In this subsection we discuss the multiplication of two $B$-bit input numbers in (polynomial) $Mul_w(B)$ time, where $w$ is the word-size. For example, for integers there exists a multiplication algorithm by Harvey and van der Hoeven \cite{harvey:hal-02070778} with time complexity $Mul_w(B)=\mathcal{O}\left(B\log{B}\right)$ (generally the running time can also depend on $w$, although in this case it does not). Let us notice that although the result is unpublished yet, we use it due to its easy to read time complexity; it is trivial to use other algorithms instead, like the one from F\"urer \cite{DBLP:journals/siamcomp/Furer09}, or the more practical one by Sch\"onhage and Strassen \cite{DBLP:journals/computing/SchonhageS71}.
We establish the complexity of multiplying $x$ $B$-bit numbers. Our divide and conquer algorithm splits the numbers into two (equal sized) groups, recursively multiplies each, and multiplies the results in $Mul_w\left(\frac{xB}{2}\right)$ time. By a direct application of the Master Theorem by Bentley et al. \cite{BentleyHJ} we prove the following lemma.

\begin{lemma} \label{lem:MultiplyX_Bbits}
~Multiplying $x$ $B$-bit numbers costs 
\begin{itemize}
\item $\mathcal{O}(Mul_w(xB)\log(xB))$ time if $Mul_w(xB) = \Theta(xB\log^k(xB))$ for some constant $k$, 
\item $\mathcal{O}((xB)^{c})$ else if $Mul_w(xB) = \mathcal{O}((xB)^c)$ for some constant $c\geq 1$,
\end{itemize}
assuming that $Mul_w(N)$ is a polynomial time algorithm that multiplies two $N$-bit numbers.
\end{lemma}
\begin{proof}
The algorithm simply splits the numbers in two equal-sized groups, recursively multiplies each, and then multiplies the results. Let $N=xB$. We have that the running time for multiplying $x$ $B$-bit numbers is $T(N)=2T(\frac{N}{2}) + Mul_w(N)$. Since $c_{crit}=\log_2{2}=1$, and $Mul_w(N)=\Omega(N)$, the Master Theorem \cite{BentleyHJ} gives two cases. Either $Mul_w(N)=\Theta(N\log^k(N))$ for some constant $k$, in which case $T(N) = \mathcal{O}(Mul_w(N)\log{N})$, or else $Mul_w(N) = \mathcal{O}(N^c)$ for some constant $c\geq 1$ (such a constant exists since we assume polynomial time multiplications). In this case, since it holds that $2Mul_w(\frac{N}{2}) \leq 2Mul_w(N)$, we get that $T(N)=Mul_w(N)$ if $c>c_{crit}=1$. Notice that we handled all cases, since $Mul_w(N)=N$ is handled by the first case with $k=0$, and whatever does not fit in the first case, definitely fits in the second, since we assumed that $Mul_w(N)$ is polynomial in $N$.
\end{proof}

\begin{corollary} \label{Multiplication}
Multiplying $x$ $B$-bit numbers costs polynomial time by using any polynomial time algorithm for multiplying two $B$-bit numbers as a black box. Especially if we use Harvey and Van Der Hoeven's algorithm, the time cost is $\mathcal{O}\left(xB\log^2{(xB)}\right)$.
\end{corollary}

Let us also notice that the way to divide two $B$-bit numbers is simply storing both the numerator and the denominator. Comparing two numbers $x_1=\frac{num_1}{den_1}$ and $x_2=\frac{num_2}{den_2}$ can be done by comparing $num_1\times den_2$ and $num_2\times den_1$. The only other operation we need when working with such fractions is subtracting a $B$-bit number $x=\frac{num}{den}$ from $1$. This is simply $\frac{den-num}{den}$.

\section{NP-Completeness} \label{sec:NP}

An NP-Completeness proof for the integer log-probability version of the WLCS problem has been given in~\cite{DBLP:journals/dam/CyganKRRW16}. This is a closely related problem, with the main difference being that products are replaced with sums. We do not know of any way to reduce from this log-probability version to WLCS other than exponentiating. As stated in the explanation of our model of computation in Section~\ref{sec:preliminaries}, there is no limit on the number of bits needed to represent a single number (it just occupies a lot of words). This means that, if the input consisted of $I$ bits, and there was a number (probability) represented with $\frac{I}{100}$ bits, exponentiating would result in a number with $2^{\frac{I}{100}}$ bits, meaning the reduction would not be a polynomial-time one. For this reason, we believe that although it is easier to prove NP-Completeness for the integer log-probability version of the problem, there is no easy way to use it for proving NP-Completeness for the general version. We, thus, give a reduction from the NP-Complete problem Subset Product \cite{DBLP:books/fm/GareyJ79} which proves NP-Completeness directly for the general problem.

Notice that for alphabets consisting of one letter, the problem is trivial since there is no uncertainty at all. In the following, we prove that even for alphabets consisting of two letters, the problem is NP-Complete.

\begin{definition}[Subset Product]
Given a set $L$ of $n$ integers and an integer $P$, find if there exists a subset of the numbers in $L$ with product $P$.
\end{definition}

\begin{lemma}
WLCS is NP-Complete, even for alphabets of size $2$.
\end{lemma}
\begin{proof}
Obviously $WLCS\in NP$ since the increasing subsequences $\pi_1, \pi_2$ and the string $s$ for which $P_X(\pi_1,s)\geq a, P_Y(\pi_2,s)\geq a$ are a certificate which, along with the input, can be used to verify in polynomial time that the problem has a solution.

Let $(L,P)$ be an instance of Subset Product and let $n=|L|$. By $L_i$ we denote the $i$-th number of the set $L$, assuming any fixed ordering of the $n$ numbers of $L$. We give a polynomial-time reduction to a $(X,Y,a,k)$ instance of WLCS, with alphabet size $2$ (we call the letters $'A'$ and $'B'$).

The core idea is the following: The weighted sequences have $n$ positions (plus $2$ more for technical reasons related to the threshold $a$). The number $k$ is equal to the length of the sequences, meaning that we pick every position, and the only question is whether we picked letter $'A'$ or letter $'B'$. Letter $'A'$ in position $i$ corresponds to picking the $i$-th number in the original Subset Product, while letter $'B'$ corresponds to not picking it. Finally, the letters $'A'$ picked in $X$ form an inequality of the form: "some product is $\geq P$", while the same letters in $Y$ form the inequality: "the same product is $\leq P$". For these two to hold simultaneously, it must be the case that we found some product equal to $P$, which is the goal of the original Subset Product.

More formally, the weighted sequences have size $n+2$. Let $c_i=\frac{1}{1+L_i}$ and $d_i=\frac{1}{1+\frac{1}{L_i}}$.
\begin{align*}
p^{(X)}_i('A')& =c_iL_i, 1\leq i\leq n &
p^{(Y)}_i('A')& =\frac{d_i}{L_i}, 1\leq i\leq n \\
p^{(X)}_{n+1}('A')& =1 &
p^{(Y)}_{n+1}('A')& = \prod_{j=1}^n{\frac{1}{L_i}} = \frac{\prod_{j=1}^n{c_i}}{\prod_{j=1}^n{d_i}} \\
p^{(X)}_{n+2}('A')& =\frac{1}{P^2} &
p^{(Y)}_{n+2}('A')& =1 \\
\end{align*}
where $p^{(X)}_i('B') = 1 - p^{(X)}_i('A')$ for all $i$, and similarly for $Y$. Notice that, in particular, $p^{(X)}_i('B') =c_i, 1\leq i\leq n$ and $p^{(Y)}_i('B') =d_i, 1\leq i\leq n$. Finally, we set $k=n+2$ and $a=\frac{\prod_{j=1}^n{c_i}}{P}$.

First of all, notice that since we must find a string of length $n+2$, we must choose a letter from every position. Thus, a choice of letter at some position on $X$ corresponds to the same choice of letter at that position on $Y$. The choice of letter on positions $n+1$ and $n+2$ is $'A'$ in both cases since \[p^{(X)}_{n+1}('B')=p^{(Y)}_{n+2}('B')=0\]
Suppose that the numbers at positions $\{i_1,\ldots,i_{\ell}\}$ give product $P$: \[\prod_{j=1}^{\ell}{L_{i_j}}=P\]
Then, we form the string $s$ by picking $'A'$ at positions $\{i_1,\ldots, i_{\ell},n+1,n+2\}$ and $'B'$ at all other positions. Thus
\[P_X(\{1,2,\ldots,n+2\},s)=\frac{\prod_{j=1}^{\ell}{L_{i_j}}\prod_{j=1}^n{c_i}}{P^2}=\frac{\prod_{j=1}^n{c_i}}{P}=a\]
\[P_Y(\{1,2,\ldots,n+2\},s)=\frac{\prod_{j=1}^n{d_i}\prod_{j=1}^n{c_i}}{\prod_{j=1}^{\ell}{L_{i_j}}\prod_{j=1}^n{d_i}}=\frac{\prod_{j=1}^n{c_i}}{P}=a\]

Conversely, suppose a solution for the WLCS problem, where the string $s$ is formed by picking $'A'$ at positions $\{i_1,\ldots, i_{\ell},n+1,n+2\}$ and $'B'$ at all other positions. It holds that:
\[P_X(\{1,2,\ldots,n+2\},s)=\frac{\prod_{j=1}^{\ell}{L_{i_j}}\prod_{j=1}^n{c_i}}{P^2}\geq a \implies \prod_{j=1}^{\ell}{L_{i_j}}\geq P\]
\[P_Y(\{1,2,\ldots,n+2\},s)=\frac{\prod_{j=1}^n{d_i}\prod_{j=1}^n{c_i}}{\prod_{j=1}^{\ell}{L_{i_j}}\prod_{j=1}^n{d_i}}\geq a\implies \prod_{j=1}^{\ell}{L_{i_j}}\leq P\]
The above imply that $\prod_{j=1}^{\ell}{L_{i_j}}=P$.
Finally, notice that all computations are done in polynomial time, due to Corollary~\ref{Multiplication}.
\end{proof}

\section{EPTAS for Bounded Alphabets, Improved PTAS for Unbounded Alphabets} \label{sec:eptas}

We now give an Efficient Polynomial Time Approximation Scheme ($EPTAS$) for the case where our alphabet size is bounded ($|\Sigma|=O(1)$). Let us notice that this is the case when working with DNA sequences ($|\Sigma|=4$), the most usual application of weighted sequences. The same algorithm is an improved (when compared to \cite{DBLP:journals/dam/CyganKRRW16}) $PTAS$ in the case of unbounded alphabets. This means that the WLCS problem is Fixed-Parameter Tractable for constant size alphabets and thus belongs to the corresponding complexity class $FPT$ as shown in Corollary~\ref{cor:FPT}.

The authors in~\cite{DBLP:journals/dam/CyganKRRW16} first noted that there is no $FPTAS$ unless $P=NP$, and so we can only hope for an $EPTAS$. Our result relies on their following result:

\begin{lemma}[Lemma 4.6 of \cite{DBLP:journals/dam/CyganKRRW16}] \label{PTAS}
It is possible to find, in polynomial time, a solution of size $d$ to the WLCS optimization problem such that the optimal value $OPT$ is guaranteed to be either $d$ or $d+1$ (however we do not know which one holds).
\end{lemma}

Their $PTAS$ uses the above result and in case the approximation is guaranteed to be good enough ($d>(1-\epsilon)(d+1)$, which implies that $d>(1-\epsilon)OPT$), it stops. Else, it holds that $\frac{1}{\epsilon}\geq d+1\geq OPT$, and the $PTAS$ exhaustively searches all subsequences of $X$, all subsequences of $Y$, and all possible strings of length $d+1$, for a total complexity of
\[\mathcal{O}\left(Mul_w\left(\frac{B}{\epsilon}\right)\log\left(\frac{B}{\epsilon}\right)|\Sigma|^{\frac{1}{\epsilon}}{\binom{n}{\frac{1}{\epsilon}}}^2\right)\]
$Mul_w(\frac{B}{\epsilon})\log(\frac{B}{\epsilon})$ is the time needed to multiply $d+1$ numbers with at most $B$-bits each, by Lemma~\ref{lem:MultiplyX_Bbits}, and is insignificant compared to the other terms.
Our $EPTAS$ improves the exhaustive search part to 
\[\mathcal{O}\left(Mul_w\left(\frac{B}{\epsilon}\right)\frac{n
}{\epsilon}|\Sigma|^{\frac{1}{\epsilon}}\right)\]
which is polynomial in the input size, in case of bounded alphabets. The following lemma is needed.

\begin{lemma}\label{GivenStringFindSubsequence}
Given a weighted sequence $X$ of length $n$, and a string $s$ of length $d$, it is possible to find the maximum number $a$ such that there exists an increasing subsequence $\pi$ of length $d$ for which $P_X(\pi, s) = a$. The running time of the algorithm is $O(Mul_w(dB)nd)$, where $B$ is the maximum number of bits needed to represent each probability in $X$.
\end{lemma}

\begin{proof}
We use dynamic programming. Let $s_j$ be the string formed by the first $j$ letters of $s$, $c _j$ be the $j$-th letter of $s$ and $opt_X(i,j)$ be the maximum number such that there exists an increasing subsequence $\pi '$ of length $j$ whose last term $\pi '_j$ is at most $i$ and for which $P_X(\pi ',s_j)=opt_X(i,j)$. Since we choose whether $c _j$ is picked from the $i$-th position of $X$, it holds that:
\[opt_X(i,j)=\max\{opt_X(i-1,j), opt_X(i-1,j-1)p^{(X)}_i(c _j)\}\]
For the base cases, $opt_X(i,0)=1$ for all $i$ (we can always form the empty string with certainty, by not picking anything), and $opt_X(0,j)=0$ for $j>0$ (not picking anything never gives us a non-empty string).
We are interested in the value $opt_X(|X|,|s|)$.
\end{proof}

Now we are ready to give our $EPTAS$.

\begin{theorem} \label{thm:eptas}
For any value $\epsilon \in (0,1]$ there exists an $(1-\epsilon)$-approximation algorithm for the WLCS problem which runs in $\mathcal{O}\left(poly(I)+\frac{n}{\epsilon}Mul_w\left(\frac{B}{\epsilon}\right)|\Sigma|^{\frac{1}{\epsilon}}\right)$ time and uses $\mathcal{O}\left(poly(I)\right)$ space, where $I$ is the input size, $n=|X|+|Y|$ and $B$ is the maximum number of bits needed to represent a probability in $X$ and $Y$. Consequently, the WLCS problem admits an $EPTAS$ for bounded alphabets.
\end{theorem}

\begin{proof}
We begin by using Lemma~\ref{PTAS} to find an $a$-subsequence of length $d$, such that the optimal solution is at most $d+1$. If $d+1\geq \frac{1}{\epsilon}$, we are done, since in that case we have a $\frac{d}{d+1}=1-\frac{1}{d+1}\geq (1-\epsilon)$ approximation. Otherwise, we try all possible strings $s \in |\Sigma|^{d+1}$, and use Lemma~\ref{GivenStringFindSubsequence} to check if any one of them can appear in both weighted sequences with probability at least $a$.
\end{proof}

\begin{corollary}\label{cor:FPT}
$WLCS \in FPT$ for bounded alphabets, parameterized by the solution length.
\end{corollary}
\begin{proof}
Follows directly from \cite{DBLP:journals/cj/Marx08},~Proposition~2.
\end{proof}

\section{No EPTAS for Unbounded Alphabets}\label{sec:noeptas}

We have already seen that there is no $FPTAS$ for WLCS, even for alphabets of size $2$, unless $P=NP$. We have also shown an $EPTAS$ for bounded alphabets and a $PTAS$ for unbounded alphabets. The natural question that arises is: Is it possible to give an $EPTAS$ for unbounded alphabets?

We answer this question negatively, by proving that WLCS is $W[1]$-hard, meaning that it does not admit an $EPTAS$ (and is in fact not even in $FPT$) unless $FPT=W[1]$ (\cite{DBLP:journals/cj/Marx08},~Corollary~$1$). To show this, we give a $2$-step $FPT$-reduction from Perfect Code, which was shown to be $W[1]$-Complete in \cite{DBLP:journals/ipl/Cesati02}, to $k$-sized Subset Product and then to WLCS. The $k$-sized Subset Product problem is the Subset Product problem with the additional constraint that the target subset must be of size $k$. 

\begin{definition}[Perfect Code]
A perfect code is a set of vertices $V'\subseteq V$ with the property that for each vertex $u\in V$ there is precisely one vertex in $N_G(u)\cap V'$, where $N_G(u)$ is the set of adjacent nodes of $u$ in $G$.
\end{definition}
In the perfect code problem, we are given an undirected graph $G$ and a positive integer $k$, and we need to decide whether $G$ has a $k$-element perfect code. Notice that the definition of a perfect code implies that there is a perfect code iff there is a set $V'\subseteq V$ for which $\bigcup_{u\in V'}{N_G(u)}=V$ and $N_G(u)\cap N_G(v)=\emptyset$ for all $u,v\in V', u\neq v$.
First we show that $k$-sized Subset Product is $W[1]$-hard.

\begin{lemma}
$k$-sized Subset Product is $W[1]$-hard.
\end{lemma}
\begin{proof}
Let $(G=(V,E),k)$ be an instance of Perfect Code. Suppose that the vertices are $V=\{1,\ldots,n\}$. First of all, we compute the first $n$ prime numbers using the Sieve of Eratosthenes. We denote the $i$-th prime number as $p_i$. The set of positive integers $L=\{L_1,L_2,\ldots,L_n\}$ as well as the positive integer $P$ are defined as follows:
\[L_v=\prod_{u\in N_G(v)}p_u,~P=\prod_{v=1}^{n}p_v\]
Notice that due to the unique prime factorization theorem, a subset of $k$ numbers from the set $L$ have product $P$ iff $G$ has a $k$-element Perfect Code.

The size of our primes is $O(n\log{n})$ due to the prime number theorem. Thus, they require $O(\log{n})$ bits to be represented. Each integer in $L$, as well as in $P$, is computed using Corollary~\ref{Multiplication} in $O(n\log^3{n})$ time, for an overall $O(n^2\log^3{n})$ complexity for our reduction. Since the new parameter $k$ is the same as the old one (no dependence on $n$), our reduction is in fact an $FPT$-reduction.
\end{proof}

Our result for this section is the following.

\begin{theorem}\label{thm:DiagonalWS}
WLCS, parameterized by the length of the solution, is $W[1]$-hard.
\end{theorem}
\begin{proof}
To prove the theorem we create diagonal weighted sequences. That is, we require each letter to appear only in one position and vice-versa. In this way, the subsequences picked for $X$ and $Y$ are the same. The above rule is broken by the addition of two auxiliary letters that are there to make the probabilities add up to $1$ in each position. This creates no problem because we make sure that these letters are never picked. Finally, we force the product to be equal to our target, by forcing it to be at most our target and at least our target at the same time.

More formally, let $(L=\{L_1,L_2,\ldots,L_n\},k,P)$ be an instance of the $k$-sized Subset Product problem and let $M=m^{k+1}$, where $m$ is the maximum number in set $L$. Notice that if $m^k\leq P$ then we only need to check the product of the highest $k$ numbers of $L$, which means the problem is solvable in polynomial time. Thus we can assume that $M\geq m^k>P$. The alphabet of $X, Y$ is $\Sigma=\{1,2,\ldots,n, n+1, n+2, n+3\}$ and we set $a=\frac{1}{PM^k}$.
\begin{flalign*}
p^{(X)}_i(i)& =\frac{L_i}{M},~1\leq i\leq n &
p^{(Y)}_i(i)& =\frac{1}{ML_i},~1\leq i\leq n \\
p^{(X)}_{n+1}(n+1)& =\frac{1}{P^2} &
p^{(Y)}_{n+1}(n+1)& =1 \\
p^{(X)}_i(n+2)& =1-p^{(X)}_i(i),~1\leq i\leq n+1 &
p^{(Y)}_i(n+3)& =1-p^{(Y)}_i(i),~1\leq i\leq n+1
\end{flalign*}
All non-specified probabilities are equal to 0. Notice that symbols $n+2$ and $n+3$ are used to guarantee that probabilities sum up to $1$.

We show that the instance $(X,Y,a,k+1)$ has a solution iff $(L,k,P)$ has a solution. Suppose there exists a solution to $(L,k,P)$. Then, there exists an increasing subsequence $\pi=(i_1,\ldots,i_k)$ such that $\prod_{j=1}^{k}{L_{i_j}}=P$. Let $\pi'$ be $\pi$ extended by the number $i_{k+1}=n+1$ and $s$ be the string $i_1i_2\ldots i_{k+1}$. It holds that $P_{X}(\pi', s)=P_{Y}(\pi', s)=a$.

Conversely, suppose there exists a solution to $(X,Y,a,k+1)$. Then there exist increasing subsequences $\pi=(i_1,\ldots,i_{k+1}), \pi'=(j_1,\ldots,j_{k+1})$ and a string $s$ such that $P_{X}(\pi, s)\geq a, P_{Y}(\pi', s)\geq a$. First of all, notice that, due to $p^{(X)}_i(n+3)=p^{(Y)}_i(n+2)=0$ for all $i$, $s$ does not contain letters $n+2$ and $n+3$, which leaves only one choice for every position. Also each letter appears only once in each sequence, and in the same position. Thus, $\pi=\pi'$, and due to our construction the $i$-th letter of $s$ is the $i$-th member of $\pi$. Finally, not picking position $n+1$ would result in $P_Y(\pi,s)<a$ due to the fact that $P<M$.
Thus, the last letter of $s$ is $n+1$. It holds that:
\[P_X(\{i_1,\ldots,i_{k+1}\},s)\geq a\implies \frac{\prod_{i=1}^{k}{L_{\pi_i}}}{P^2M^k}\geq \frac{1}{PM^k}\implies \prod_{i=1}^{k}{L_{\pi_i}}\geq P\]
\[P_Y(\{i_1,\ldots,i_{k+1}\},s)\geq a\implies \frac{1}{M^k\prod_{i=1}^{k}{L_{\pi_i}}}\geq \frac{1}{PM^k}\implies \prod_{i=1}^{k}{L_{\pi_i}}\leq P\]
The above two inequalities imply a $k$-sized subset of $L$ with product equal to $P$.

The reduction is a polynomial-time one, due to Corollary~\ref{Multiplication}. More than that, it is an $FPT$-reduction since the new parameter $k$ is equal to the old parameter incremented by one, and thus has no dependence on $n$.
\end{proof}

\section{Matching Conditional Lower Bound on any PTAS}\label{sec:lowerbound}
In the $d$-SUM problem, we are given $N$ numbers and need to decide whether there exists a $d$-tuple that sums to zero. Patrascu and Williams \cite{DBLP:conf/soda/PatrascuW10} proved that any algorithm for solving the $d$-SUM problem requires $n^{\Omega(d)}$ time, unless the Exponential Time Hypothesis ($ETH$) fails. To show this, they first proved a hardness result for a variant of 3-SAT, the sparse 1-in-3 SAT.

\begin{definition}[Sparse 1-in-3 SAT]
Given a boolean formula with $n$ variables and $O(n)$ clauses in 3 CNF form, where each variable appears in a constant number of clauses, determine whether there exists an assignment of the variables such that each clause is satisfied by exactly one variable.
\end{definition}

They first prove the following hardness result under $ETH$.

\begin{proposition}\label{PatrascuProposition}
Under $ETH$, there is an (unknown) constant $s_3$ such that there exists no algorithm to solve sparse 1-in-3 SAT in $\mathcal{O}(2^{\delta n})$ time for $\delta <s_3$.
\end{proposition}

By assuming an $n^{\mathcal{O}(d)}$ time algorithm for $d$-SUM they disproved the above fact, which cannot happen under $ETH$. We use the same technique for proving an $n^{\Omega(k)}$ lower bound for $k$-sized Subset Product.

\begin{lemma}\label{LowerBoundSubsetProduct}
Assuming the $ETH$, the problem of $k$-sized Subset Product cannot be solved in $\mathcal{O}(n^{\frac{s_3k}{101}})$ time on instances satisfying $k<n^{0.99}$ and each number in the input set $L$ has $\mathcal{O}\left(\log{n}(\log{k}+\log{\log{n}})\right)$ bits, where $n$ is the size of $L$, and $P$ is the target which can be arbitrarily big.
\end{lemma}
\begin{proof}
Let $f$ be a sparse 1-in-3 SAT instance with $N$ variables and $M=\mathcal{O}(N)$ clauses, and $k>\frac{1}{s_3}$. Conceptually, we split the variables of $f$ into $k$ blocks of equal size - apart from the last block that may have smaller size. Each block contains at most $\lceil \frac{N}{k} \rceil$ variables, and thus there are at most $2^{\lceil \frac{N}{k} \rceil}$ different assignments of values to the group-of-variables within a block. For each block and for each one of these assignments we generate a number which serves as an identifier of the corresponding block and assignment. Thus, there are $n=k2^{\lceil \frac{N}{k} \rceil}$ different identifiers.

Let $p_i$ be the $i$-th prime number. In order to compute an identifier, we initialize it to $p_b$, where $b$ is the index of the identifier's corresponding block. Then, we run through all of the $M=\mathcal{O}(N)$ clauses and do the following: suppose we process the $i$-th clause and let $0\leq j\leq 3$ be the number of variables of the identifier's corresponding assignment that satisfy the clause. We update the identifier by multiplying it with $p_{k+i}^j$.

Since each variable appears only in a constant number of clauses, each identifier is a product of $\mathcal{O}(\frac{N}{k})$ numbers. The prime number theorem guarantees $\mathcal{O}(\log{N})$ bits to represent each factor, which means the identifiers have $\mathcal{O}(\frac{N}{k}\log{N})$ bits. Using the fact that $n=k2^{\lceil \frac{N}{k} \rceil}$, each identifier is represented by $\mathcal{O}\left(\log{n}(\log{k}+\log{\log{n}})\right)$ bits.

These $n$ identifiers, along with the target $P=\prod_{i=1}^{k+M}p_i$ (recall that $p_i$ is the $i$-th prime number), form a $k$-sized Subset Product instance. This preprocessing step costs $\mathcal{O}(2^{\frac{N}{k}})$ time, ignoring polynomial terms, which is more efficient than $\mathcal{O}(2^{s_3N})$.

Due to the unique prime factorization, a solution to the $k$-sized Subset Product corresponds to a solution in $f$ and vice-versa. If the running time of the $k$-sized Subset Product was $\mathcal{O}(n^{\frac{s_3k}{101}})$ then we could solve the above instance in $\mathcal{O}((k2^{\frac{N}{k}})^{\frac{s_3k}{101}})$ time.

Since $k=\frac{n}{2^{\lceil \frac{N}{k} \rceil}}$ and $k<n^{0.99}$, it follows that $\frac{n}{2^{\lceil \frac{N}{k} \rceil}}<n^{0.99}\implies n^{0.99}<2^{99\lceil \frac{N}{k} \rceil}$. But $k<n^{0.99}$, which means $k<2^{99\lceil \frac{N}{k} \rceil}$.

Thus the previous running time becomes $\mathcal{O}(2^{\frac{100}{101}s_3N})$. Both the preprocessing step and the solution of the $k$-sized Subset Product can be achieved in time $\mathcal{O}(2^{\delta N})$, where $\delta <s_3$. However, this would violate Proposition~\ref{PatrascuProposition}.
\end{proof}

Using the above, we are ready to prove our (matching) lower bound, conditional on $ETH$.

\begin{theorem}
Under $ETH$, there is no $PTAS$ for WLCS with running time $|I|^{o(\frac{1}{\epsilon})}$, where $|I|$ is the input size in bits.
\end{theorem}
\begin{proof}
Suppose that such an algorithm $A(I,\epsilon)$ existed. Let $R()$ be the polynomial time reduction from $k$-sized Subset Product to WLCS given in the proof of Theorem~\ref{thm:DiagonalWS}. Then, there is a solution to $k$-sized Subset Product iff there is a solution to WLCS of size $k+1$, or, equivalently, iff the optimal solution to WLCS is at least $k+1$.

Using the hypothetical $A(I,\epsilon)$ with an appropriate value of $\epsilon$, we solve $k$-sized Subset Product more efficiently than possible, thus reaching a contradiction.

Consider the following algorithm for $k$-sized Subset Product, where there are $|L|$ numbers in the input, each having $\mathcal{O}\left(\log{|L|}(\log{k}+\log{\log{|L|}})\right)$ bits and $k<|L|^{0.99}$. Given an instance $(L,k,P)$, we define the instance for the WLCS to be $I=R(L,k,P)$. We run $A(I,\frac{1}{2(k+1)})$ and if the output is at least $k+1$ we return that $(L,k,P)$ is satisfied, otherwise we return that it cannot be satisfied.

Note that if $k$-sized Subset Product is solvable, then $OPT(I)\geq k+1$, and the value output by $A$ is at least $(1-\frac{1}{2(k+1)})(k+1)=k+\frac{1}{2}>k$. Thus, the value output by $A$ is at least $k+1$. On the other hand, if $k$-sized Subset Product is not solvable, then $OPT(I)< k+1$, and obviously the value output by $A$ is at most k.

Thus we found an algorithm for $k$-sized Subset Product whose running time is $|I|^{o(k)}$. Since $I$ is obtained by a polynomial time reduction, its size is bounded by a polynomial in $|(L,k,P)|$. Therefore, the above running time becomes $|(L,k,P)|^{o(k)}$. Under our assumptions, this becomes $|L|^{o(k)}$, which is not feasible under $ETH$, due to Lemma~\ref{LowerBoundSubsetProduct}.
\end{proof}

\section{Conclusion} \label{sec:conclusion}

In this paper we prove NP-Completeness for the WLCS decision problem, and give a $PTAS$ along with a matching conditional lower bound for the optimization problem. In the most usual setting, where the alphabet size is constant, the above $PTAS$ is in fact an $EPTAS$, and it is known that no $FPTAS$ can exist unless $P=NP$. In the Appendix we give a transformation such that algorithms for the WLCS problem can also be applied for the $(a_1,a_2)$-WLCS problem.

In proving that WLCS does not admit any $EPTAS$, we proved that it is $W[1]-hard$. It may be interesting to determine the exact complexity of WLCS in the $W-hierarchy$.

\bibliography{WLCS}

\appendix
\section{One Threshold is Enough} \label{app_sec:one_threshold}

For clarity purposes, some proofs and technical discussions are moved in this appendix. In particular, in this section we show that $(a_1,a_2)$-WLCS and WLCS are equivalent, thus one threshold is enough. Furthermore, we show that the rounding errors introduced by working with reals (logarithms and roots) may cause a similar algorithm from a paper by Cygan et al. \cite{DBLP:journals/dam/CyganKRRW16} to err if standard rounding is used.

In the following, $B$ corresponds to the maximum number of bits to represent a number in the input (a probability or a symbol of the alphabet). $B$ is not to be confused with the word-size $w$ since an input number may need many words to be represented.

\begin{lemma}
Given an instance $(X,Y,a_1,a_2,k)$ of $(a_1,a_2)$-WLCS $(a_1 < a_2)$, it is possible to reduce it to an instance $(X',Y',a,k+1)$ of WLCS. The construction of $X'$ and $Y'$ requires $\mathcal{O}(n|\Sigma|Mul_w(B))$ time, while parameter $a$ is computed in $\mathcal{O}(Mul_w(nB)\log{(nB)})$ time, where $n=|X|+|Y|$ is the total length of the weighted sequences $X$ and $Y$, while $B$ is the maximum number of bits needed to represent an input number.
\end{lemma}
\begin{proof}
We first provide a sketch of the proof. Our goal is to use the same weighted sequences with one additional position at the end. We introduce a new letter ($'\%'$) which only appears in this position, and we make sure that any correct algorithm picks it, by making its probability very appealing (high). Since we cannot assign a probability higher than one, increasing it is simulated by reducing all other probabilities, in all positions. Knowing that this specific letter is picked at this specific position allows us to choose the two corresponding probabilities in a way that completes the proof. In order for the probabilities to sum to $1$ in every position, we introduce two auxiliary letters ($'\#'$ and $'\$'$) that are never picked ($'\$'$ never appears on the first weighted sequence, $'\#'$ never appears on the second).

The alphabet $\Sigma '$ of $X', Y'$ is the alphabet $\Sigma$ of $X, Y$ extended by three new letters, $\Sigma '=\Sigma \cup \{'\#', '\$', '\%'\}$. Let $m=\frac{a_1}{2}$ and $a=m^k a_1$. Notice that since $k\leq n$, the size of $a$ in bits is only polynomial compared to the input size, not exponential. The new sequences $X'$ and $Y'$ are constructed as follows:
\begin{flalign*}
p^{(X')}_i(\sigma)& =mp^{(X)}_i(\sigma), 1\leq i\leq |X|, \sigma \in \Sigma &
p^{(Y')}_i(\sigma)& =mp^{(Y)}_i(\sigma), 1\leq i\leq |Y|, \sigma \in \Sigma \\
p^{(X')}_i('\#')& =1-\sum_{\sigma \in \Sigma ' \setminus \{'\#'\}}
p^{(X')}_i(\sigma), \forall{i} &
p^{(Y')}_i('\$')& =1-\sum_{\sigma \in \Sigma ' \setminus \{'\$'\}}p^{(Y')}_i(\sigma), \forall{i} \\
p^{(X')}_{|X|+1}('\%'&) =1 &
p^{(Y')}_{|Y|+1}('\%'&) =\frac{a_1}{a_2}
\end{flalign*}
All non-specified probabilities are equal to $0$.

If there exists a solution to $(X,Y,a_1,a_2,k)$, then there exist two increasing subsequences $\pi_1=(i_1,\ldots,i_k), \pi_2=(j_1,\ldots,j_k)$ and a string $s$ such that $P_X(\pi_1, s)\geq a_1, P_Y(\pi_2, s)\geq a_2$. Define $\pi_1'=(i_1,\ldots,i_k,|X|+1), \pi_2'=(j_1,\ldots,j_k,|Y|+1)$ and $s'$ to be equal to $s$ extended with the letter $'\%'$. It holds that:

$P_{X'}(\pi_1', s')=m^kP_X(\pi_1,s)\geq m^ka_1=a, P_{Y'}(\pi_2', s')=m^kP_Y(\pi_2,s)\frac{a_1}{a_2}\geq m^ka_2\frac{a_1}{a_2}=a$

Conversely, suppose there exists a solution to $(X',Y',a,k+1)$. Then, there exist increasing subsequences $\pi_1=(i_1,\ldots,i_{k+1}), \pi_2=(j_1,\ldots,j_{k+1})$ and a string $s$ such that $P_{X'}(\pi_1, s)\geq a, P_{Y'}(\pi_2, s)\geq a$. First of all, notice that, due to $p^{(X')}_i('\$')=p^{(Y')}_i('\#')=0$ for all $i$, $s$ does not contain letters $'\$'$ and $'\#'$. In addition, the letter $'\%'$ only appears at the last position, and it is the only possible option for this position. Finally, the last position shall be used on both subsequences, because otherwise $P_{X'}(\pi_1,s), P_{Y'}(\pi_2,s)\leq m^{k+1}< a$. Thus, the last letter of $s$ is $'\%'$. If we denote by $s'$ the string $s$ without its last letter, it holds that $P_X(\{i_1,\ldots,i_k\},s')\geq a_1, P_Y(\{j_1,\ldots,j_k\},s')\geq a_2$.

The computation of $a$ requires $\mathcal{O}(Mul_w(nB)\log{(nB)})$ time due to Corollary~\ref{Multiplication}, and the $n|\Sigma|$-multiplications of two numbers with at most $B$ bits each cost $\mathcal{O}(n|\Sigma|Mul_w(B))$. All other computations take linear time.
\end{proof}

We note that \cite{DBLP:journals/dam/CyganKRRW16} proved the same result, but their reduction required computations with real numbers (raising to the $\log_{a_2}{a_1}$ power). To the best of our knowledge, there is no way to modify that reduction so that it tolerates the rounding error in the word $RAM$ introduced by working with roots and logarithms.

In what follows, we show that the rounding errors may cause the algorithm by Cygan et al. \cite{DBLP:journals/dam/CyganKRRW16}, which reduces any instance of WLCS to a more restricted class of instances, to err.  This does not rule out the possibility that more clever rounding algorithms (depending on the input size) may indeed be used so that the algorithm does not err; however we are not aware of any such rounding technique, and even if it exists, the algorithm would probably become too complicated compared to ours.

\begin{lemma}
The reduction from $(a_1,a_2)$-WLCS to WLCS with only one threshold given by Cygan et al. in \cite{DBLP:journals/dam/CyganKRRW16} may err, if exact computations with logarithms and roots are not assumed (assuming the rounding technique does not depend on the input, for example it only keeps a constant number of decimal digits).
\end{lemma}
\begin{proof}
We prove the above with an example that demonstrates that the rounding error, introduced by not assuming exact computations with logarithms and roots, may cause the reduction to err.

Let $a_1=\frac{1}{8}, a_2=\frac{1}{4}$ and the two weighted sequences $X$ and $Y$ on alphabet $\Sigma=\{a,b\}$ be: 

\begin{center}
\begin{tabular}{|l|l|l|l|l|}
\hline
$X$ & 1 & 2 & 3 & 4 \\ \hline
a & 1 & 1 & 1 & $\frac{1}{8}$ \\ \hline
b & 0 & 0 & 0 & $\frac{7}{8}$ \\ \hline
\end{tabular}
\quad
\begin{tabular}{|l|l|l|l|l|}
\hline
$Y$ & 1 & 2 & 3 & 4 \\ \hline
a & $x$ & $\frac{1}{2}$ & $\frac{1}{2}$ & 1 \\ \hline
b & $1-x$ & $\frac{1}{2}$ & $\frac{1}{2}$ & 0 \\ \hline
\end{tabular}
\end{center}

\noindent where $0\leq x\leq 1$ is a constant to be specified later.
For $x=1$, the weighted $LCS$ is $aaaa$ and for $x<1$ the weighted $LCS$ is $aaa$. The transformation described in~\cite{DBLP:journals/dam/CyganKRRW16} would give $a=\frac{1}{8}, \gamma=\frac{3}{2}$ and the new sequences would be:

\begin{center}
\begin{tabular}{|l|l|l|l|l|}
\hline
$X'$ & 1 & 2 & 3 & 4 \\ \hline
a & 1 & 1 & 1 & $\frac{1}{8}$ \\ \hline
b & 0 & 0 & 0 & $\frac{7}{8}$ \\ \hline
\# & 0 & 0 & 0 & 0 \\ \hline
\end{tabular}
\quad
\begin{tabular}{|l|l|l|l|l|}
\hline
$Y'$ & 1 & 2 & 3 & 4 \\ \hline
a & $x^{\gamma}$ & $\frac{1}{2}^{\gamma}$ & $\frac{1}{2}^{\gamma}$ & 1 \\ \hline
b & $(1-x)^{\gamma}$ & $\frac{1}{2}^{\gamma}$ & $\frac{1}{2}^{\gamma}$ & 0 \\ \hline
\# & $1-x^{\gamma}-(1-x)^{\gamma}$ & $1-2*\frac{1}{2}^{\gamma}$ & $1-2*\frac{1}{2}^{\gamma}$ & 0 \\ \hline
\end{tabular}
\end{center}
Since $\frac{1}{2}^{\gamma}$ is an irrational number, it is rounded to some number $r=\left\lfloor \frac{1}{2}^{\gamma}\right \rceil$. Suppose $r<\frac{1}{2}^{\gamma}$. In this case, when $x=1$, while the weighted $LCS$ is $aaaa$ the algorithm returns $aaa$ due to the rounding errors. On the other hand, if $r>\frac{1}{2}^{\gamma}$, we can always find an appropriate $x<1$ such that the weighted $LCS$ should have been $aaa$ but the algorithm returns $aaaa$ due to the rounding errors. To show this, let $x=\left(\frac{k-1}{k}\right)^{2}$ for some integer $k$. Then $x^{\gamma}=\left(\frac{k-1}{k}\right)^{3}$. It holds that $\left(\frac{k-1}{k}\right)^{3}r^2$ is an increasing function of $k$ which converges to $r^2>\frac{1}{8}$. Thus, we can find a big enough $k$ such that $x^{\gamma} r^2\geq \frac{1}{8}$ and err on this particular example, as long as the rounding technique does not depend on the input (for example it only keeps a constant number of decimal digits).
\end{proof}
Once again, the above is not a proof that the algorithm given by Cygan et al. can never be correct, despite of the rounding algorithm used. It just shows that it is necessary to explicitly specify such a rounding algorithm in order to construct a correct algorithm.

\end{document}